\documentclass{IEEEtran}

\usepackage{amssymb, amsmath, amsfonts, amsthm, pzccal}
\usepackage{graphicx, cite, mathtools}
\usepackage{dsfont}
\usepackage{algorithmic}
\usepackage{bbm}
\usepackage{textcomp}
\usepackage{xcolor}
\usepackage{tikz, standalone}
\usepackage[short]{optidef}
\usepackage[ruled,vlined,linesnumbered]{algorithm2e}
\usepackage{setspace}

\newcommand\SNR{\mathrm{SNR}}

\theoremstyle{plain}
\newtheorem{lem}{Lemma}
\newtheorem{prop}{Proposition}
\newtheorem{thm}{Theorem}
\newtheorem{cor}{Corollary}

\normalsize
\allowdisplaybreaks

\begin{document}
	
	\title{A Markov Chain Approach for Myopic Multi-hop Relaying: Outage and Diversity Analysis}
	
	\author{Andreas Nicolaides, \textit{Student Member, IEEE,}
		Constantinos Psomas, \textit{Senior Member, IEEE,}
		and~Ioannis Krikidis, \textit{Fellow, IEEE}
		\thanks{This work has received funding from the European Research Council (ERC) under the European Union’s Horizon 2020 research and innovation programme (Grant agreement No. 819819) and from the Marie Skłodowska-Curie project PAINLESS under the European Union’s Horizon 2020 research and innovation programme (Grant agreement No. 812991).}
		\thanks{A. Nicolaides, C. Psomas and I. Krikidis are with the IRIDA Research Centre for Communication Technologies, Department of Electrical and Computer Engineering, University of Cyprus, Cyprus, e-mail:\{anicol09, psomas, krikidis\}@ucy.ac.cy. Part of this work was presented at the IEEE Global Communications Conference, Taipei, Taiwan, December 2020 \cite{nicolaides2020}.}}
	
	\maketitle

	\begin{abstract}
		In this paper, a cooperative protocol is investigated for a multi-hop network consisting of relays with buffers of finite size, which may operate in different communication modes. The protocol is based on the myopic decode-and-forward strategy, where each node of the network cooperates with a limited number of neighboring nodes for the transmission of the signals. Each relay stores in its buffer the messages that were successfully decoded, in order to forward them through the appropriate channel links, based on its supported communication modes. A complete theoretical framework is investigated that models the evolution of the buffers and the transitions at the operations of each relay as a state Markov chain (MC). We analyze the performance of the proposed protocol in terms of outage probability and derive an expression for the achieved diversity-multiplexing tradeoff, by using the state transition matrix and the related steady state of the MC. Our results show that the proposed protocol outperforms the conventional multi-hop relaying scheme and the system's outage probability as well as the achieved diversity order depend on the degree of cooperation among neighboring nodes and the communication model that is considered for every relay of the network.
	\end{abstract}
	
	
	\IEEEpeerreviewmaketitle
	
	\section{Introduction}
		
		The initial deployment and standardization of 5G networks has recently emerged, while research in academia and industry has already made some first steps towards the beyond 5G (B5G)/6G-era \cite{saad2020}. Nowadays, there is a continuous need for an increasing number of applications and devices to be connected, leading to an ever-growing network of nodes that have to communicate. Cooperative networks with multiple relays that can assist the transmission of information from a source to a destination is an appealing technology that enables device-to-device (D2D) communications, due to its implementation simplicity and easy scalability \cite{tehrani2014}. Especially in B5G/6G networks, research needs to conceive innovative ideas to satisfy the challenging demands of ultra-reliable low-latency communications for massive connectivity networks \cite{chen2021}. Cooperative relay communications is a promising low-cost solution with high flexibility, which can help the connection of trillions of devices with enhanced reliability and low energy consumption. As such, relay communications can be considered for the implementation of local and private 5G networks, as they provide an adaptive physical layer and flexible transmission protocol. Due to their potential advantages towards future generations of wireless communications, cooperative relay networks have attracted considerable research interest and have been considered for several practical applications, such as in wireless ad-hoc networks \cite{atallah2016}, millimeter-wave (mmWave) communications \cite{rois2015}, underwater and air-to-ground networks \cite{saeed2019,chen2018}, and secrecy communications \cite{wang2020}.
		
		Numerous studies have investigated the performance of cooperative networks with multiple relays over a single transmission path (multi-hop relaying). More specifically, the end-to-end performance of a dual-hop network was analyzed in \cite{hasna2004}, while \cite{bjornson2013} extended the classical dual-hop relaying to a generalized model with hardware impairments. Other studies extended this approach to multi-hop schemes. In \cite{hasna2003}, the outage probability performance of a multi-hop system over Nakagami fading channels was studied. In addition, the authors in \cite{jamali2015} presented a new protocol for half-duplex multi-hop relaying networks based on the concept of buffer-aided relaying and investigated the corresponding achievable rates. Multi-hop relaying has been also proposed to assist connectivity for mmWave communications \cite{rois2015,lin2015}, a technology that is considered as one of the main components of 5G networks, but is highly susceptible to path blockage. In particular, \cite{rois2015} investigated a joint scheduling and congestion control policy in multi-hop mmWave networks, in order to maximize the throughput under fairness requirements. On the other hand, the connectivity of mmWave networks with multi-hop relaying was analyzed in \cite{lin2015} by considering a stochastic geometry approach, where the obstacles are modeled as a Boolean model.  
		
		Cooperative diversity is another relaying technique that has received a lot of attention in the literature, as it enables broadcast transmission and spatial diversity of the participating nodes. In the seminal work in \cite{laneman2004}, Laneman \textit{et al.} proposed several techniques of cooperative communication, such as selection relaying and incremental relaying, and investigated their outage performance. Moreover, the authors in \cite{ribeiro2005} studied general cooperating setups, consisting of multiple transmission paths that include an arbitrary number of cooperating hops, and derived asymptotic expressions for the average symbol error probability. These setups consisted of either a single relay or multiple relays in parallel transmission paths. 
		
		A cooperation scenario for multi-hop networks was introduced in \cite{boyer2004}, where it has been shown that the spatial diversity gain could be achieved by combining at each node the signals that have been concurrently sent by all the preceding terminals along a single transmission path. Based on this idea, in \cite{sadek2007}, a class of cooperative diversity protocols was proposed, where each relay combines the signals received from an arbitrary number of previous nodes. The authors proved that this class of cooperative protocols can achieve the same diversity gain as \cite{boyer2004}. Significant diversity gain can be also achieved, according to \cite{dong2012}, if we consider a multi-hop buffer-aided system, where every relay has a buffer of sufficient size and at each time-slot a stored packet is transmitted over the best hop, based on the received signal-to-noise ratio ($\SNR$). Furthermore, in \cite{sreeram2012}, the maximum diversity and multiplexing gain, as well as the achieved diversity-multiplexing tradeoff (DMT) of various multi-hop cooperative network topologies are characterized. However, full cooperation of the nodes in such networks exhibits a number of practical difficulties in its implementation, especially in terms of multi-node coordination and power management. To overcome these issues, the authors in \cite{ong2008} proposed the \textit{myopic} decode-and-forward (DF) strategy as an information theory concept. In this strategy, each node of the network cooperates with a limited number of subsequent neighboring nodes. They showed that the achievable rate increases considerably, while the complexity of its implementation remains low.
		
		In several cases, the devices that need to be connected in local and private 5G networks may support different transmission policies affecting their communication capabilities e.g., for energy conservation. Such devices could be utilized to provide further optimized services in order to improve the networks' performance, based on the available resources. The authors in \cite{niyato2007} presented a queuing model for the performance analysis of several sleep and wakeup strategies in a network with solar-powered wireless sensors. In \cite{medepally2010}, a relay selection scheme was considered for a cooperative network with energy harvesting (EH) relays that can assist the communication from source to destination only if they have sufficient energy. It was observed that, the overall performance of such networks depends on the EH parameters of the relays. A similar scenario was considered in \cite{luo2013}, where a new relay selection method was proposed, based on the throughput gain of the EH relays with enough stored energy, which improved the overall performance of the cooperative network. Moreover, in \cite{morsi2018}, the on-off transmission policy for a buffer aided EH node was studied in terms of outage probability and average throughput, where the EH node transmits information only if the stored energy exceeds a certain value, otherwise it remains silent.
		
		It is, therefore, an important and challenging problem to understand how the utilization of relaying nodes with different communication capabilities affects the performance of a multi-hop network with limited cooperation. Motivated by this, in this paper, we propose a general cooperative protocol over a multi-hop network, where the relays have buffers of finite size and may operate in different communication modes. The protocol is inspired by the myopic DF strategy \cite{ong2008}, and can be applied to networks with an arbitrary number of relays with different modes of operation. For example, such scenario could be considered for wireless networks, where some intermediate nodes are crucial for the sustainability of the communication and so they are connected to the power grid and are always able to transmit information. The remaining nodes could be self-powered through EH and could transmit data only if the harvested energy is above a required level \cite{medepally2010,luo2013}. The main contributions of the paper are summarized as follows:
		\begin{itemize}
			\item A novel myopic-based cooperative protocol over a multi-hop network is proposed. Through this paper, we extend the work presented in \cite{ong2008} by studying the performance of myopic strategy in terms of outage probability and diversity gain. To our knowledge, no previous work in the literature analyzes myopic strategy from a communication theory perspective. A system model is presented, where the flow of information is assisted by using finite buffers at each relay of the network. Finally, the myopic-based protocol is extended to the case where relays may operate in different communication modes. Therefore, a fundamental approach of how the flow of information from source to destination can be conveyed is presented, based on the status of the buffers and the communication capabilities of each relay.
			\item For the analysis of the system in terms of outage probability, we model the evolution of the considered network as a state Markov chain (MC), by taking into account the transitions that take place at the buffers and the communication modes of the relays. This approach provides a flexible and elegant modeling of the different instances that the network can be found. By using the state transition matrix and the related steady state of the MC, we investigate a complete theoretical framework for the performance analysis of such cooperative networks. The presented framework is general and can be adapted to an arbitrary number of relays, any myopic strategy and several communication strategies supported by the participating relays.
			\item Our results demonstrate that as the number of cooperating nodes increases, the performance of the system is enhanced both in terms of outage probability and diversity gain. Furthermore, it is shown that the diversity order that can be achieved by the proposed protocol depends significantly on the communication strategy that is supported by every relay of the network. Finally, by extending the proposed protocol for multi-branch networks, the outage probability of the system is improved, while the overall diversity gain depends on the diversity gain that each branch can separately achieve. 
		\end{itemize}
		\noindent As such, the proposed myopic protocol can provide useful insights for the design and realization of local and private 5G networks with reduced computational complexity and memory requirements and increased energy efficiency.
		
		The remainder of the paper is organized as follows. Section \ref{sys_model} introduces the system model and Section \ref{strategy} describes the implementation of the proposed protocol. A state Markovian model approach used for the derivation of the system's outage probability is presented in Section \ref{markov}. In Section \ref{analysis}, we provide the numerical expressions of the outage probability analysis and the DMT for the proposed protocol. Our numerical and simulation results are presented in Section \ref{results} and finally, some concluding remarks are stated in Section \ref{conc}.
	
		\textit{Notation:} Lower and upper case boldface letters denote vectors and matrices, respectively; $\mathbb{P}[X]$ denotes the probability of the event $X$ and $\mathbb{E}[X]$ represents the expected value of $X$; $\mathds{1}_{X}$ is the indicator function, where $\mathds{1}_{X}=1$ if $X$ is true, otherwise $\mathds{1}_{X}=0$; $\Im(x)$ returns the imaginary part of $x$ and $\jmath=\sqrt{-1}$ denotes the imaginary unit; $\Phi(\cdot)$ is the cumulative distribution function (cdf) of the standard normal distribution and $O(\cdot)$ denotes the big $O$ notation; $[x]^{+}=\max(0,x)$, $\lceil x \rceil =\min\{m\in\mathbb{Z}|m\geq x\}$, $\binom{n}{k}=\frac{n!}{k!(n-k)!}$ and $(2n-1)!!=(2n-1)(2n-3)\ldots3\cdot1$.\vspace{-2mm}

	\section{System Model}\label{sys_model}\vspace{-1mm}
	\subsection{Network topology}
	A wireless network topology is considered, which consists of a single source $S$, $N$ intermediate relays $R_{1}, R_{2}, \ldots, R_{N}$, and a single destination $D$\footnote{A network topology with multiple sources and destinations can be also considered, which is left for future work.}. For ease of notation, we let node $i$ correspond to the relay $R_{i}$, $1 \leq i \leq N$, and nodes $0$ and $N+1$ correspond to the source $S$ and destination $D$, respectively. At the relays, the DF scheme is employed for forwarding the signals. Moreover, time is assumed to be slotted and $x(t)$ is used to denote the signal that $S$ sends to $R_{1}$ at time-slot $t$ with normalized energy, i.e. $\mathbb{E}[|x(t)|^2]=1$. Each transmitter (the source $S$ or a relay $R_{i}$) transmits with the same fixed power $P$. The destination $D$ receives data based on a $k$-hop myopic DF strategy \cite{ong2008}, $1 \leq k \leq N+1$, where $k$ represents the maximum number of nodes that a transmitter can forward data to. More specifically, node $i$ ($0 \leq i \leq N$) can send data to $L_{i}=\min (k,N-i+1)$ subsequent nodes. As such, at each time-slot, the $i$-th transmitter splits its power to $L_{i}$ partitions. Therefore, a signal is sent through the link $i \rightarrow j$ with transmit power $a_{i,j}P$, where $a_{i,j}$ denotes the power splitting parameter, such that $\sum_{j=i+1}^{L_{i}+i} a_{i,j}=1$. However, at a given time-slot $t$, only the successfully decoded signals can be forwarded to the appropriate nodes. For that reason, each relay $R_{i}$ has a data buffer (data queue) $b_{i}$ of finite size $L_{i}$, where it can store the decoded signals for forwarding\footnote{Note that each node sends data to the subsequent $L_{i}$ nodes \textit{concurrently}. Therefore, a storage space of the same size is required in order to hold the data that will be forwarded.}, based on the proposed protocol described in Section \ref{strategy}. An example of this topology is presented in Fig.~\ref{fig:model}, for $k=2$ hops and $N=3$ relays.\vspace{-2mm}

	\subsection{Channel model} 
	For the analysis, we consider independent and identically distributed (i.i.d.) channel links that experience propagation path loss, which is assumed to follow the power law $d_{i,j}^{-\eta}$, where $d_{i,j}$ is the distance between the nodes $i$ and $j$ and $\eta$ denotes the path loss exponent. Without loss of generality, we assume the ordering $d_{i,j}<d_{i,j+1}$, $\forall$ $i,j$, $i<j\leq N+1$. Note that this assumption corresponds to distance-based routing protocols in multi-hop networks that take into consideration the Euclidean distance among the nodes e.g. the shortest-path-routing \cite{wang2015}. Furthermore, all wireless links exhibit fading which is modeled as frequency-flat Rayleigh block fading\footnote{Even though we consider Rayleigh fading, the proposed analytical framework is general and the extension to other fading models is straightforward as we only need to consider their probability distributions.}. This signifies that the fading coefficients $h_{i,j}$ remain constant during one time-slot, but change independently for different time-slots, by following a circularly symmetric complex Gaussian distribution with zero mean and unit variance i.e., $h_{i,j}\sim\mathcal{CN}(0,1)$. We assume that during one time-slot, the relays of the system can receive and transmit data simultaneously i.e., they operate in full-duplex mode. Since we focus on the performance of the myopic DF scheme, we consider an ideal scenario where we ignore the self-interference and interference from other relays; interference mitigation can be achieved through sophisticated signal processing and equalization techniques \cite{boyer2004,korpi2017}. In other words, the presented analysis serves as a communication theory bound.\vspace{-2mm}
	
	\begin{figure}[t!]\centering
		\includegraphics[width=1\linewidth]{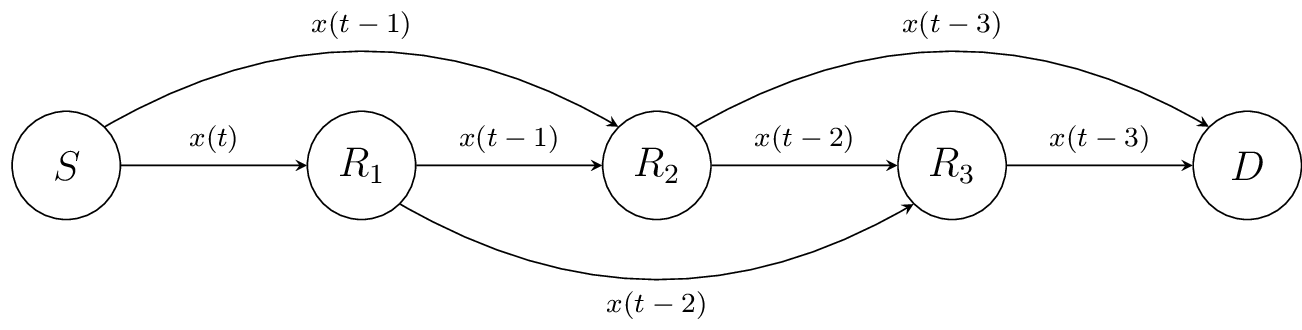}\vspace{-3mm}
		\caption{Topology for $2$-hop myopic DF strategy in a wireless network with three relays. The arrows indicate all the available links in the system and for each link the corresponding signal that will be transmitted is shown.}
		\label{fig:model}\vspace{-3mm}
	\end{figure}
	
	\subsection{Relay communication model}
	The relays of the network operate based on two different communication modes, namely the \textit{active} and \textit{silent} modes. If a relay is active during one time-slot, then it is able to simultaneously transmit and receive the decoded signals. In the silent mode, the relay can still receive data but is unable to forward its decoded signals to the following nodes \cite{niyato2007}.
	We consider a scenario where a subset of the participating relays are always active. The remaining relays are able to switch their operations between the two communication modes, and therefore can be either active or silent. Such relays could correspond to devices that are always connected to the power grid and switch off their transmitters occasionally by design, in order to conserve energy. Alternatively, the switching between the two communication modes could occur probabilistically, based on the EH profile of the devices \cite{morsi2018}. 
	
	We assume that during a time-slot, a relay's communication mode remains constant, but can change independently between the two modes for different time-slots. In the considered scenario, the relay's decision for the communication mode at each time-slot is modeled as a Bernoulli distribution with parameter $q$, which indicates the probability that a relay operates in silent mode. For the relays that are able to transmit and receive data at every time-slot, or equivalently are always in the active mode, the probability of operating in silent mode is equal to $q=0$. By using the Bernoulli distribution, we provide a unified analysis for various decision-based scenarios at the relays, such as EH scenarios where the decision is related to the status of the energy buffers \cite{li2016}. However, it is important to note that the proposed analytical framework is general and any other distribution can be applied.
	
	\section{A $k$-hop Myopic-based Protocol} \label{strategy}
	We now present our proposed protocol based on the $k$-hop myopic DF strategy \cite{ong2008}. The protocol describes the procedure over which the flow of information from $S$ to $D$ is conveyed within the considered network. Based on the presented system model, at time-slot $t$ the received signal at the $j$-th receiver is given by 
	\begin{equation}
		\label{rec_signal}
		y_{j}(t) = \sum_{i=0}^{N}\mathds{1}_{i\rightarrow j}\sqrt{d_{i,j}^{-\eta}\alpha_{i,j}P}h_{i,j}x(t-j+1)+w_{j}(t),
	\end{equation}
	where $w_j(t) \sim \mathcal{CN}(0,\sigma^2)$ is the additive white Gaussian noise (AWGN) with variance $\sigma^2$ and $\mathds{1}_{i\rightarrow j}$ equals to one if node $i$ transmits a signal to node $j$ at the current time-slot, otherwise it is equal to zero. In this work, we focus on the fundamental analysis of the myopic scheme, so a perfect (global) channel state information (CSI) is assumed and each receiver is able to combine the received signals coherently through co-phasing\footnote{In practice, a channel estimation process is required and the associated overhead increases as $k$ increase; note that $k$ can be adapted to the available resources to achieve a balance between estimation and system performance.}. Therefore, the corresponding instantaneous $\SNR$ at the $j$-th receiver during one time-slot is given by
	\begin{equation}
		\label{rec_SNR}
		\SNR_{j} = \dfrac{P}{\sigma^{2}}\left(\sum_{i=0}^{N}\mathds{1}_{i\rightarrow j}|h_{i,j}|\sqrt{d_{i,j}^{-\eta}a_{i,j}}\right)^2,
	\end{equation}
	where $|h_{i,j}|$ is a random variable that follows a Rayleigh distribution with unit scale parameter. We assume that $\nu$ relays of the network, $0\leq \nu \leq N$, have dual-mode communication capabilities, while the remaining $N-\nu$ relays are always active. Moreover, $S$ is able to transmit all its available data at every time-slot. A signal at the $j$-th receiver is successfully decoded if the instantaneous $\SNR_{j}$ is not less than a predefined threshold $\gamma$ i.e., $\SNR_{j}\geq \gamma$, otherwise an outage occurs. Each relay keeps in its buffer the signals that were successfully decoded, in order to forward them through the appropriate channel links. For this, the buffer of each relay is used as an one-dimensional array with indexed elements, where the element $b_{i}[n]$, $1\leq n \leq L_{i}$, corresponds to the $n$-th most recent signal that the $i$-th relay receives.
	
	At each time-slot, the network performs three specific tasks: $(i)$ transmission of information, $(ii)$ buffer shifting, and $(iii)$ information decoding. More specifically, at an arbitrary time-slot $t$, the network follows the procedure described below:
	\begin{itemize}
		\item \textit{\underline{Transmission of information:}} At first, the node $S$ sends the signals $x(t-j+1)$ to the nodes $j=1, \dots, k$, respectively. Then, every dual-mode relay determines whether it will operate in active or silent mode. Recall that the relays with a single-mode operation are considered to be always active. The $i$-th relay of the network, $i=1,\ldots, N$, forwards the signal $x(t-j+1)$ to the corresponding node $j=i+1, \dots, \min (i+k, N+1)$, if and only if, it is in active mode and its buffer element $b_{i}[j-i]$ is not empty i.e., a signal is stored in this element. Thus, at the specific time-slot, the $j$-th node of the network can receive the signal $x(t-j+1)$ simultaneously from at most $k$ previous nodes. The example in Fig.~\ref{fig:model} shows an instance of the network, where at time-slot $t$ all the relays are active and every buffer's element has a previously decoded signal. As a result, at the transmission phase, all the available links of the system, for the considered 2-hop myopic DF strategy, are used.
		\item \textit{\underline{Shifting operation:}} After the transmission phase, each relay prepares its buffer for the next time-slot. In particular, it shifts the elements one position to the right. In other words, $b_{i}[2]$ will get the data of $b_{i}[1]$, $b_{i}[3]$ the data of $b_{i}[2]$, etc. Therefore, the first element of each buffer becomes unassigned for the decoding of $x(t-j+1)$ (see decoding process below). The shifting operation is required for the proper transmission of the appropriate signals to the corresponding nodes at each time-slot.
		\item \textit{\underline{Decoding process:}} Every relay combines all the received signals $x(t-j+1)$, acquired at the transmission phase, and attempts to decode the message. The first element of its buffer $b_{i}[1]$ is used for the outcome of the decoding process: it stores the signal if it is successfully decoded, otherwise it becomes an empty element. It is important to note that, an empty element indicates that the decoding of the signal failed (i.e. similarly to a pointer indicating a null value) and thus is not able to be forwarded. Finally, the destination $D$ combines the received signals $x(t-N)$ and if the message is not successfully decoded i.e., $\SNR_{N+1}<\gamma$, then the system is in outage and the message is lost.
	\end{itemize}
	According to the above procedure, it is clear that some of the elements at the buffers might be empty due to the shifting operation. Consequently, the proper transmission of the successfully decoded signals to the appropriate nodes relies on the communication mode of each relay and the content of its buffer that evolves with time. Based on these features, in the following section, we introduce a theoretical framework that is exploited for the system's outage probability analysis.

\section{A State Markovian Model Approach} \label{markov}
	For the analysis of the considered system, we provide a theoretical framework that models the evolution of the relays' buffers and the communication mode of the dual-mode relays as a MC. In this section, the state transition matrix construction and the derivation of the stationary distribution of the MC are presented, which will be used for the computation of the system's outage probability in Section \ref{analysis}.\vspace{-2mm}
	
\subsection{Definition of MC states} 
	Firstly, the MC states are defined in order to represent the different instances that characterize the status of the network. The distinction between these instances depends on the buffers' evolution with time and the transitions at the dual-mode communication strategy. A state of the MC, or equivalently a \textit{network state}, needs to capture the transitions that take place at the relays' buffers and the dual-mode relays' operations. Thus, it is necessary to represent the evolution of the buffers and the relays' status by dividing the network states into two separate sub-states, namely the \textit{buffer states} and the \textit{relay states}.
	
	Recall from Section \ref{strategy} that, as a result of the decoding process and the shifting operation, each buffer's elements can be found in two possible conditions: either to have the $n$-th most recently received signal or to be empty. We denote the $m$-th buffer state by
	\begin{equation}
	\label{buf_st}
	u_{m}\triangleq(\beta_{1,m} \: \beta_{2,m} \: \ldots \: \beta_{N,m}),
	\end{equation}
	where $\beta_{i,m}$ is a one-dimensional binary array associated with the $i$-th relay, for which each element $\beta_{i}[n]$ equals $0$ if the corresponding buffer element $b_{i}[n]$ is empty, otherwise is equal to $1$. Therefore, each array $\beta_{i,m}$ indicates the non-empty elements of the corresponding buffer $b_{i}$ at the buffer state $u_{m}$, while from the definition it is apparent that its size equals $L_{i}$. The number of the buffer states is given by all the possible combinations that can be derived by each $\beta_{i}$'s element. Thus, each buffer state is a vector of finite size $L_{\beta}=\sum_{i=1}^{N}L_{i}=(2N-k+1)\frac{k}{2}$, that represents which elements in each buffer have decoded signals, and the total number of buffer states is equal to $M_{B}=2^{L_{\beta}}$.

	However, the buffers' evolution and equivalently the transition between different instances of the network depends also on the communication mode of each dual-mode relay that varies for each time-slot. In the considered scenario, the dual-mode relays are predetermined and are given by the ordered set $H=\{i|R_i \text{ is dual-mode}\}$. The $m$-th relay state is denoted as 
	\begin{equation}
	\label{relay_st}
	v_{m}\triangleq(\lambda_{H(1),m} \: \lambda_{H(2),m} \: \ldots \: \lambda_{H(\nu),m}),
	\end{equation}
	where $\lambda_{H(i),m}$ indicates the communication mode of the $i$-th dual-mode relay at the $m$-th state and is equal to $0$ if the relay operates in  silent mode, otherwise equals $1$. Note that this parameter is defined only for the dual-mode relays of the network, since the other relays are always active and the operations related to their communication capabilities do not vary with time. Thus, each relay state is a $1\times\nu$ vector that captures all the possible transitions at the dual-mode relays' communication modes and depicts which relays are active at each instance of the network. The total number of relay states is then derived as $M_{R}=2^\nu$.
	
	The network states are constructed by the concatenation of the buffer and relay states that were previously defined. Therefore, the $m$-th network state, $1\leq m\leq M$, is denoted as
	\begin{align}
		\label{net_st}
		s_{m}&\triangleq(v_{m} \: u_{m})\nonumber\\
		&=(\lambda_{H(1),m} \ldots \: \lambda_{H(\nu),m} \: \beta_{1,m} \ldots \: \beta_{N,m}),
	\end{align}
	which represents the joint status of the relays' buffers and the dual-mode relays' communication activity. The total number of network states that are considered for the MC is given by all the possible pairs of buffer and relay states that can be derived and is equal to $M=M_{B}M_{R}=2^{L_{\beta}+\nu}$. Since this concatenation results in a binary representation of a decimal number, the MC states are predefined and arranged in a numerical ascending order, such that the states $s_{1}$ and $s_{M}$ are denoted by the $1\times (L_{\beta}+\nu)$ binary vectors $(00\ldots0)$ and $(11\ldots1)$, respectively.
	
	\subsection{State transition matrix and stationary distribution}
	The state transition matrix is a square matrix containing information on the transition probabilities between the states of the MC. Specifically, it defines how the system evolves with time and indicates which of the available links will be used at each time-slot. Let $\mathbf{A}$ denote the $M \times M$ state transition matrix of the MC, where the entry $p_{l,m}=\mathbb{P}(s_{m}\rightarrow s_{l})=\mathbb{P}(X_{t+1}=s_{l}|X_{t}=s_{m})$ is the probability of the transition from network state $s_{m}$ at time $t$ to state $s_{l}$ at time $(t+1)$. The calculation of these probabilities relies on the communication mode of each dual-mode relay and the status of each relay's buffer, and consequently from the corresponding parameters $\lambda_{i}$ and $\beta_{i}$.

\setlength{\textfloatsep}{4pt}
	\begin{algorithm}[t!]
		\label{algo}
		\SetAlgoLined
		\KwIn{\\$N\leftarrow$ Number of relays\\$H\leftarrow$ Set of dual-mode relays\\$k\leftarrow$ Number of hops}
		\KwOut{\\$\mathbf{A}\leftarrow$ Transition Matrix}
		Compute $L_{i}=\min(k,N+1-i), i=1,\ldots,N$\\
		Compute $L_{\beta}=\left( 2N-k+1\right)\dfrac{k}{2}$\\
		Compute $M_{B}=2^{L_{\beta}}$, $M_{R}=2^{\nu}$, $M=M_{B}M_{R}$\\
		\For{$m=1$ to $M$}{
			\For{$l=1$ to $M$}{
				Assume $s_{m}\rightarrow s_{l}$ exists\\
				\For{$i=1$ to $N$}{
					Get $\beta_{i,m}$, $\beta_{i,l}$\\
					\For{$n=1$ to $L_{i}-1$}{
						\If{$\beta_{i,m}[n] \neq \beta_{i,l}[n+1]$}{
							$s_{m}\rightarrow s_{l}$ does not exist
						}
					}
				}
				Compute $p_{l,m}$ using \eqref{trans_prob}
			}
		}
		\caption{Generation of the state transition matrix.}
	\end{algorithm}
	
	Algorithm \ref{algo} shows the proposed procedure for the construction of the state transition matrix $\mathbf{A}$, given the number of relays $N$, the ordered set of dual-mode relays $H$ and the number of hops $k$. First of all, we derive the size of each array $\beta_i$ and the size of the buffer states. Moreover, we compute the number of buffer and relay states and consequently the number of states of the MC i.e., (lines $1-3$). Then, we need to detect all the possible transitions between the states. Since the decision of the communication mode at each dual-mode relay at time-slot $t$ is independent from what was decided at previous time-slots, the only parameters that affect the validity of a state transition are the arrays $\beta_{i}$, $1\leq i \leq N$, due to their shifting operation. Thus, for each pair of states $ (s_{m},s_{l}) $, we examine if the variations at the elements of each array $\beta_{i}$ are consistent with the evolution of the buffers as in the proposed protocol. Each array $\beta_{i,m}$ can be extracted by isolating the buffer state $u_{m}$ from the network state $s_{m}$ and taking all the elements of the resulting array from the index $u_{m}[\sum_{j=1}^{i-1}L_{j}+1]$ to the index $u_{m}[\sum_{j=1}^{i}L_{j}]$. A transition from $s_{m}$ to $s_{l}$ exists, if the contents of all the arrays shift one position to the right. This is equivalent to the equality $\beta_{i,m}[n] = \beta_{i,l}[n+1]$, $1 \leq n \leq L_{i}-1$, for all the arrays $\beta_{i}$, $i=1,\ldots,N$. If at least one of these equalities does not hold, then the transition is not possible. The first element of each array at the new buffer state $ \beta_{i,l}[1] $ indicates if an outage occurred at the corresponding relay, while the elements of the new relay state indicate which dual-mode relays are active. As the decoding is handled separately by each relay, these probabilities are independent and therefore the transition probability is given as their product. The entries of the state transition matrix are then given by
	\begin{equation}
		\label{trans_prob}
		p_{l,m}=
		\begin{cases}
		&\hspace{-3mm}\displaystyle \prod_{i=1}^{\nu} \displaystyle\prod_{j=1}^{N}\bigg[\lambda_{H(i),l}(1-q) +(1-\lambda_{H(i),l})q\bigg]\\
		&\hspace{-3mm}\times\bigg[ \beta_{j,l}[1] (1-P_{o}(j,m))+\left(1-\beta_{j,l}[1]\right)P_{o}(j,m)\bigg],\\
		&\hspace{-3mm} \hspace{4.5cm}\text{if } s_{m}\rightarrow s_{l} \text{ exists};\\
		&\hspace{-3mm}0, \text{otherwise},
		\end{cases}
	\end{equation} 
	where $P_{o}(j,m)$ is the probability that the $j$-th node has an outage event, given that the network instance is derived by the state $s_{m}$. Note that, due to the two possible values that the elements $\lambda_{H(i),l}$ and $\beta_{j,l}[1]$ can take, the aggregate number of all possible transitions from every state are $2^{\nu+N}$. The analytical expressions for the outage probability are given in Section \ref{analysis}.
	
	We are now able to derive the stationary distribution of the MC, which is denoted as $\boldsymbol{\pi}$. In this case, the interpretation of the stationary distribution gives an insight to the long-term use of the available channel links in the system, as it indicates how the signals are being transmitted across the relays until they reach the final destination. The calculation of the steady states is given in the following Lemma \ref{steady}.
	
	\begin{lem}\label{steady}
		The state transition matrix $\mathbf{A}$ of the defined MC has a unique stationary distribution, which is given by
		\begin{equation}
			\label{steady_state}
			\boldsymbol{\pi} = (\mathbf{A}-\mathbf{I}+\mathbf{B})^{-1}\mathbf{b},
		\end{equation}
		where $\boldsymbol{\pi}$ is the stationary distribution, $\mathbf{b}=(1\: 1\: \ldots \:1)^{T}$ and $\mathbf{B}_{l,m}=1$, $\forall$ $l,m$.
	\end{lem}

	\begin{proof}
		See Appendix \ref{proofA}.
	\end{proof}
	In the next section, we provide our main results for the performance analysis of our proposed protocol. 

\section{Outage Probability \& Diversity Analysis}\label{analysis}
	Based on the obtained stationary distribution of the MC, we can now analyze the performance of the proposed protocol, in terms of outage probability and diversity gain. Firstly, we provide an expression of the outage probability at each receiver and then the system's outage probability is derived. To conform with the above notation, we assume that $S$ has an array $\beta_{0}$ of finite size $k$, in which all the elements are equal to one, since $S$ always sends a signal to the $k$ subsequent nodes. 
	
	In general, the $i$-th transmitter sends a signal to the $j$-th receiver of the network, if the corresponding indicator function $\mathds{1}_{(i\rightarrow j),m}$ is equal to one and it is calculated as
	\begin{equation}
		\label{ind_fun}
		\mathds{1}_{(i\rightarrow j),m}=
		\begin{cases}
			\beta_{i,m}[j-i]\lambda_{i,m}, & i \in H;\\
			\beta_{i,m}[j-i], & \text{otherwise}.
		\end{cases}
	\end{equation}
	Thus, the number of nodes that transmit a signal to the $j$-th receiver at state $s_{m}$ is equal to
	\begin{equation}
		\label{card}
		C_{j,m}=\sum_{i=[j-k]^{+}}^{j-1}\mathds{1}_{(i\rightarrow j),m},
		\end{equation}
	which can be at most $k$, based on the presented protocol. Since the signals are transmitted coherently, the outage probability achieved at the $j$-th node is given as follows.
	\begin{thm}\label{out_prob_thm}
		The probability of having an outage event at the $j$-th node is
		\begin{multline}
			\label{out_prob_node}
			P_{o}(j,m)=\dfrac{1}{2}-\dfrac{1}{\pi}\int_{0}^{\infty}\dfrac{1}{t}\Im \bigg[\exp\left(-\jmath t \sqrt{\dfrac{\gamma \sigma^{2}}{P}}\right)\\
			\times\prod_{i=[j-k]^{+}}^{j-1}\phi_{i,j}(t)^{\mathds{1}_{(i\rightarrow j),m}}\bigg]dt,
		\end{multline}
		\noindent where $\phi_{i,j}(t)$ is the characteristic function of $|h_{i,j}|$ and is equal to
		\begin{equation}
			\label{char_fun}
			\phi_{i,j}(t)=1+\jmath t\sqrt{\dfrac{2\pi a_{i,j}}{d_{i,j}^{\eta}}}\exp\bigg(-\dfrac{a_{i,j}t^2}{2d_{i,j}^{\eta}}\bigg)\Phi\bigg(\jmath t\sqrt{\dfrac{a_{i,j}}{d_{i,j}^{\eta}}}\bigg).
		\end{equation}
	\end{thm}
	\begin{proof}
		See Appendix \ref{proofB}.
	\end{proof}
	\noindent It is clear that for the case where the $j$-th node receives a signal from only one transmitter $i$, the outage probability can be given as the cdf of an exponential distribution i.e.,
	\begin{equation}
		\label{cdf_exp}
		P_o(j,m)=1-\exp\left(-\dfrac{\gamma \sigma^{2}d_{i,j}^{\eta}}{a_{i,j}P}\right).
	\end{equation}
	
 	In the following proposition, we provide a closed form approximation of the derived outage probability expression, based on the small argument approximation (SAA) \cite{hu2005}.
 	\begin{prop}\label{saa_prop}
 		Under the SAA, the outage probability at the $j$-th receiver is approximated by
 		\begin{equation}
 			\label{out_prob_saa}
 			P_{o}(j,m)\approx 1-\exp\bigg(-\dfrac{\gamma \sigma^{2}}{2P\theta_{j,m}}\bigg)\sum_{c=0}^{C_{j,m}-1}\dfrac{1}{c!}\bigg(\dfrac{\gamma \sigma^{2}}{2P\theta_{j,m}}\bigg)^c,
 		\end{equation}
 	
 		where
 		
 		\begin{equation}
	 		\theta_{j,m}=\frac{1}{C_{j,m}}[(2C_{j,m}-1)!!]^{1/C_{j,m}}\sum_{i=[j-k]^{+}}^{j-1}\frac{a_{i,j}}{d_{i,j}^{\eta}}\mathds{1}_{(i\rightarrow j),m}.
 		\end{equation}
 	\end{prop}
 
 	\begin{proof}
 		From \eqref{rec_SNR} it is observed that the channel gain at each receiver consists of a weighted sum of i.i.d. Rayleigh random variables. For the distribution of the weighted Rayleigh sum, the following inequality holds \cite{hitczenko1998}
 		\begin{align}
 		\label{Ray_sum_ineq}
 		P_{o}(j,m)&=\mathbb{P}\left[\sum_{i=[j-k]^{+}}^{j-1} |h_{i,j}|\sqrt{\dfrac{a_{i,j}}{d_{i,j}^{\eta}}}\mathds{1}_{(i\rightarrow j),m}<\sqrt{\dfrac{\gamma \sigma^{2}}{P}}\right]\nonumber\\
 		&\geq \mathbb{P}\left[\sum_{i=[j-k]^{+}}^{j-1} |h_{i,j}|\mathds{1}_{(i\rightarrow j),m}<\sqrt{\dfrac{C_{j,m}\gamma \sigma^{2}}{g_{j,m}P}}\right],
 		\end{align}
 		where $g_{j,m}=\sum_{i=[j-k]^{+}}^{j-1}a_{i,j}d_{i,j}^{-\eta}\mathds{1}_{(i\rightarrow j),m}$. Note that the right-hand side of \eqref{Ray_sum_ineq} follows the distribution of a normalized Rayleigh sum. The approximated expression follows by taking the right-hand side of the inequality and then by using the SAA to the cdf of the normalized sum as in \cite{hu2005}.
 	\end{proof}

	\noindent The outage probability of the system $ P_{out}(\gamma)$ can be calculated by using the steady state of the MC along with the probability of an outage event at the destination. Thus, $P_{out}(\gamma)$ can be expressed as\vspace{-1mm}
	\begin{equation}
		\label{out_prob_sys}
		P_{out}(\gamma)=\sum_{m=1}^{M}\boldsymbol{\pi}_{m}P_{o}(N+1,m).
	\end{equation} 

	From the derived expressions for the achieved outage probability, it can be seen that the system's performance depends on the number of relays ($N$), the number of hops ($k$) and the set of relays that have dual-mode operations ($H$). In order to provide more insights on the performance of such networks, we need to explore in greater detail how these parameters affect the overall outage probability. In particular, in the next section the performance of our protocol is investigated in the high $\SNR$ regime and the DMT is derived.\vspace{-2mm}
	
	\subsection{Diversity-multiplexing tradeoff}\label{dmt_sect}
	In this section, we use the presented outage expressions to derive a tradeoff between the diversity and multiplexing gains for the proposed protocol. In general, a channel achieves multiplexing gain $\rho$ and a corresponding diversity gain $\delta^{*}(\rho)$, if the target data rate $R(P)\sim\rho\log P$ and the outage probability $P_{out}(P)$ satisfy the conditions \cite{zheng2003}\vspace{-1mm}

	\begin{equation} 
		\label{mult_div_gain}
		\lim_{P\rightarrow \infty}\dfrac{R(P)}{\log P}=\rho \quad \text{and} \quad
		\lim_{P\rightarrow \infty}-\dfrac{\log P_{out}(P)}{\log P}=\delta^{*}(\rho).
	\end{equation}
	Below, we provide the proposition which characterizes the DMT that the proposed protocol can achieve.

	\begin{thm}\label{dmt_thm}
		The DMT achieved by the proposed protocol for the considered multi-hop network is given by\vspace{-1mm}
		\begin{equation}
			\label{dmt_gen}
			\delta^{*}(\rho)=(1-\rho)\min_{\substack{1\leq m\leq M \\ \boldsymbol{\pi}_{m}\nrightarrow 0}} \mathcal{L}_{m},\: \rho\in[0,1],
		\end{equation}
		where $\mathcal{L}_{m}$ denotes the minimum number of disjoint paths from $S$ to $D$ that can be obtained for the network instance derived by the state $s_{m}$.
	\end{thm}

	\begin{proof}
		See Appendix \ref{proofC}.
	\end{proof}
	
	\begin{cor}
		\label{cor_div}
		The maximum diversity order of the proposed protocol for a given network topology is achieved when $\rho=0$ and is equal to $\delta^{*}(0)=\displaystyle\min_{\substack{1\leq m\leq M \\ \boldsymbol{\pi}_{m}\nrightarrow 0}} \mathcal{L}_{m}$.
	\end{cor}
	
	We can observe from Corollary \ref{cor_div} that the maximum diversity order can vary between zero and $k$ i.e., $0\leq\delta^{*}(0)\leq k$. The exact behavior of the network in the high $\SNR$ regime depends significantly on which relays within the network are dual-mode. Below, we examine how the diversity order of the protocol is affected in different topology scenarios. 
	
	\subsubsection{Only dual-mode relays ($\nu=N$)}
	In this scenario, all the relays of the network are dual-mode and so they can transmit information only when they are active. Even though at high $\SNR$ each node is able to decode all of its received signals, when all the transmitters of a node are concurrently silent during one time-slot, the node does not receive any signal and so the relays' buffers may still have empty elements. The outage performance of the network is then dominated by the terms for which the destination does not receive information from any relay. Therefore, the system's performance converges to an outage floor value, unless there is a direct link from $S$ to $D$ i.e., $k=N+1$, which results in diversity order equal to one. For $k\leq N\leq2k$, we state the following proposition.
	\begin{prop}\label{floor_prop}
		The outage floor value of the considered network for $k<N\leq2k$ is given by
		\begin{equation}\label{floor_a}
			e(N,k)=q^{k}e(N-2,k-1)+(1-q^{k})e(N-1,k),
		\end{equation}
		and for $N=k$ is equal to
		\begin{equation}\label{floor_b}
			e(N,k)=q^{k}.
		\end{equation}
	\end{prop}
	\begin{proof}
		See Appendix \ref{proofD}.
	\end{proof}
	\noindent For $N>2k$ the outage floor value can be also extracted by following a similar approach. However, the expressions of these cases are more complex and their exact derivation is out of the scope of this paper.

	\subsubsection{Only active relays ($\nu=0$)}
	In contrast to the previous scenario, in this case, all the relays are always active and so they are able to transmit their decoded signals at every time-slot. By considering $P\rightarrow \infty$, we notice that the outage probability at each relay and for every received signal converges to zero. This implies that the stationary distribution of the transition matrix is given by $\boldsymbol{\pi}_{m}\rightarrow 0$, $m=1,\ldots,M-1$, and $\boldsymbol{\pi}_{M}\rightarrow 1$, as all the buffers are full. If we follow the same procedure as in the proof of Theorem \ref{dmt_thm} we conclude that for the specific scenario the DMT of the investigated model is equal to 
	\begin{equation}
		\label{dmt_active}
		\delta^{*}(\rho)=k(1-\rho),\: \rho\in[0,1].
	\end{equation}
	and for $\rho=0$ it can achieve a maximum diversity order equal to the number of hops\footnote{The proof of the maximum diversity order that is achieved for the specific scenario can be found in \cite{nicolaides2020}.} i.e., $\delta^{*}(0)=k$.

	\subsubsection{Deployment strategy for $0<\nu<N$}
	In the previous two scenarios we presented the two extreme cases that a network topology can be found, regarding the number of dual-mode relays. Any intermediate scenario, where only part of the relays are dual-mode i.e., $0<\nu<N$, is expected to have a performance that lies between these limits. As previously stated, the selection of which relays will have dual-mode operations can significantly affect the network's performance. The proposed protocol can achieve diversity gain, if a signal can be received from $D$ through transmission paths which consist only of active relays, since the outage probability in these instances will not depend on $q$ and it will converge to zero. Specifically, a network can achieve diversity order $\delta^{*}(0)=\epsilon$, $1\leq \epsilon \leq k-1$, if the number of dual-mode relays satisfies the condition\vspace{-1mm}
	\begin{equation}
		k-\epsilon\leq\nu\leq \bigg\lceil\dfrac{N}{k}\bigg\rceil(k-\epsilon),
	\end{equation}
	and each subset of $k$ consecutive relays contains a maximum number of $k-\epsilon$ dual mode relays.

	\subsection{Multi-branch multi-hop network}
	In general, the adaptation of the network deployment, according to the previous cases, is not always achievable or the dual-mode operation may refer to conditions for which our intervention is not feasible. To overcome this issue, the aforementioned framework can be also generalized to multi-branch networks \cite{ribeiro2005}. Specifically, we consider a cooperative system of $Z$ orthogonal branches with common source and destination nodes, where each branch consists of $N$ intermediate relays $R_i^z$, $1 \leq i \leq N$, $1 \leq z \leq Z$. For each branch a $k_{z}$-myopic DF strategy is employed independently, where $1 \leq k_{z} \leq N+1$. Moreover, each branch may choose its $\nu_{z}$ dual-mode relays differently. An example of this topology is shown in Fig.~\ref{fig:model_mult}.
	
	\begin{figure}[t!]\centering
		\includegraphics[width=0.95\linewidth]{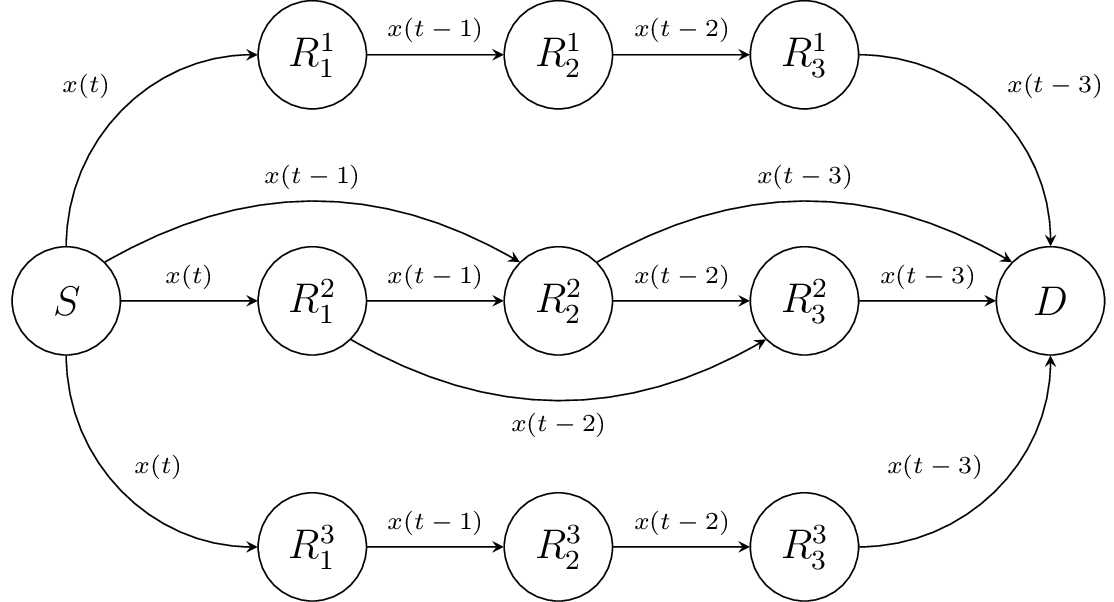}\vspace{-3mm}
		\caption{Topology of a wireless network with three branches and three relays at each branch. The branches employ the $k_{z}$-hop myopic DF strategy independently, where $k_1=k_3=1$ and $k_2=2$.}
		\label{fig:model_mult}
	\end{figure}

	We assume that the destination $D$ receives the transmitted signals by performing the selection combining (SC) technique \cite{goldsmith2005}. In this case the combiner chooses the signal of the branch with the highest $\SNR$. Since the branches in this case are orthogonal, the previous outage probability analysis can be performed at each branch separately. Therefore, the system's outage probability for the multi-branch scenario is given by	
	\begin{equation}
		\label{out_prob_mult}
		P_{out}(Z,\gamma)=\prod_{z=1}^{Z}P_{out}^{z}(\gamma)=\prod_{z=1}^{Z}\sum_{m_z=1}^{M_z}\boldsymbol{\pi}_{m_z}P_{o}(N+1,m_z),
	\end{equation}
	where $P_{o}(\cdot)$ is given by \eqref{out_prob_node}. Note that if the number of hops is the same for all the participating branches, i.e. $k_{z}=k$, $1\leq z \leq Z$, and each branch has equal number of dual-mode relays at the same position, then the outage probability of the system can be rewritten as
	
	\begin{equation}
		\label{out_prob_same_k}
		P_{out}(Z,\gamma)=\left[\sum_{m=1}^{M}\boldsymbol{\pi}_{m}P_{o}(N+1,m)\right]^Z.
	\end{equation}
	
	Regarding the high $\SNR$ regime, the following proposition provides the DMT for the case of multi-branch networks.
	\begin{prop}\label{dmt_mult_prop}
		The DMT of the proposed protocol for the considered multi-branch multi-hop network is given by
		\begin{equation}\label{dmt_mult}
		\delta^{*}(\rho)=(1-\rho)\sum_{z=1}^{Z}\min_{\substack{1\leq m_{z}\leq M_{z} \\ \boldsymbol{\pi}_{m_{z}}\nrightarrow 0}}\mathcal{L}_{m_{z}},\: \rho\in[0,1].
		\end{equation}
	\end{prop}
	\noindent By considering $P\rightarrow \infty$ and since we assume that the branches are orthogonal, the DMT is derived from \eqref{mult_div_gain} as
	\begin{equation}
		\delta^{*}(\rho)=-\lim_{P\rightarrow \infty}\frac{\log \displaystyle\prod_{z=1}^{Z}P_{out}^{z}(\gamma)}{\log P}=-\sum_{z=1}^{Z}\lim_{P\rightarrow \infty}\frac{\log P_{out}^{z}(\gamma)}{\log P},
	\end{equation}	
	which follows by the logarithmic identity $\log(xy)=\log(x)+\log(y)$. The final expression in \eqref{dmt_mult} is derived by calculating the DMT of each branch separately, following the result of Theorem \ref{dmt_thm}. Based on the above results, the maximum diversity order that can be achieved by a multi-branch network is the sum of the number of hops at each branch i.e.,
	\begin{equation}
		\label{div_gain_mult}
		\delta^{*}(0)=\sum_{z=1}^{Z}k_{z},
	\end{equation}
	which is achieved if each branch topology has only active relays i.e., $\nu_{z}=0$ $\forall$ $z=1,\ldots,Z$. Furthermore, by considering the case with an equal number of hops at all branches, the maximum diversity order can be calculated by $\delta^{*}(0)=Zk$. It is therefore easily observed that the protocol can achieve a performance enhancement of order equal to the number of branches. Even for the case where each branch's performance reaches an outage floor, the utilization of a multi-branch scenario can significantly decrease the overall outage floor value for the considered network.
	
	\subsection{Illustrative example ($N=2$ relays, $k=2$ hops)}
	We provide an example of the proposed framework that refers to a network topology with $N=2$ relays and $k=2$ hops. In this case, the size of the buffers $b_{1}$ and $b_{2}$ (and therefore of the arrays $\beta_{1}$ and $\beta_{2}$) is $L_{1}=2$ and $L_{2}=1$, respectively. The concatenation of the arrays $\beta_{1}$ and $\beta_{2}$ results in a binary vector of finite size $L_{\beta}=3$, which represents a buffer state of the MC. In this example, there are $M_{B}=2^{3}=8$ different buffer states. The set of dual-mode relays for the considered network can be any subset of the ordered set $H=\{1,2\}$. Due to space limitations, the system's outage probability is calculated for the case where all the relays are always active i.e., $H=\emptyset$. Therefore, the resulting MC states are defined only by the buffer states of the network, and are presented in Fig. \ref{fig:states}. By following the procedure in Algorithm \ref{algo}, the state transition matrix $\mathbf{A}$ is derived as
	
	\begin{equation}
		\label{ex_trans_mat}
		\mathbf{A} =
		\begin{pmatrix}
			p_{1,1} & p_{1,2} & p_{1,3} & p_{1,4} & 0 & 0 & 0 & 0 \\
			p_{2,1} & p_{2,2} & p_{2,3} & p_{2,4} & 0 & 0 & 0 & 0 \\
			0 & 0 & 0 & 0 & p_{3,5} & p_{3,6} & p_{3,7} & p_{3,8} \\
			0 & 0 & 0 & 0 & p_{4,5} & p_{4,6} & p_{4,7} & p_{4,8} \\
			p_{5,1} & p_{5,2} & p_{5,3} & p_{5,4} & 0 & 0 & 0 & 0 \\
			p_{6,1} & p_{6,2} & p_{6,3} & p_{6,4} & 0 & 0 & 0 & 0 \\
			0 & 0 & 0 & 0 & p_{7,5} & p_{7,6} & p_{7,7} & p_{7,8} \\
			0 & 0 & 0 & 0 & p_{8,5} & p_{8,6} & p_{8,7} & p_{8,8} \\
		\end{pmatrix},
	\end{equation}
	where the entries $p_{l,m}$ denote the probabilities of the existent transitions and are given by
	\begin{multline}
		\label{ex_prob}
		p_{l,m}=
		\bigg[ \beta_{1,l}[1](1-P_{o}(1,m))+(1-\beta_{1,l}[1])P_{o}(1,m)\bigg]\\
		\times\bigg[ \beta_{2,l}[1](1-P_{o}(2,m))+\left(1-\beta_{2,l}[1]\right)P_{o}(2,m)\bigg].
	\end{multline}
	
	\begin{figure}[t!]\centering
		\includegraphics[width=1\linewidth]{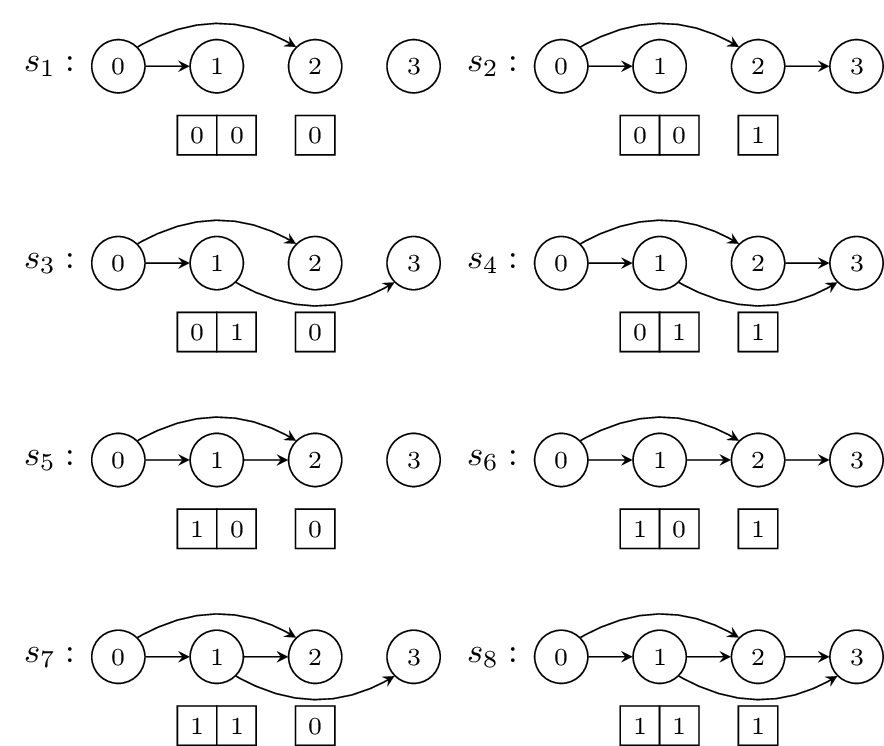}
		\caption{The possible states of the MC for a network topology with $N=2$ relays, $k=2$ hops and $H=\emptyset$. For each state, the contents of the arrays $\beta_{1}$ and $\beta_{2}$ are presented below the corresponding relay, as well as which of the available links are used.}
		\label{fig:states}
	\end{figure}

	By using \eqref{ind_fun} and from the network instances in Fig. \ref{fig:states}, we observe that $S$ always transmits a signal to $R_{1}$, while $R_{2}$ always receives from $S$ and at states $s_{5}$, $s_{6}$, $s_{7}$ and $s_{8}$ it also receives from $R_{1}$. Finally, the transmitters of $D$ are $R_{1}$ and $R_2$ at states $s_4$ and $s_8$, only $R_1$ at states $s_3$ and $s_7$ and only $R_2$ at states $s_2$ and $s_6$. At states $s_1$ and $s_5$ the destination does not receive any data, so it is on outage with probability one. From the derived results the following relations are obtained
	
	\begin{equation}
		\begin{split}
			P_{o}(1,1)&=P_{o}(1,2)=\ldots=P_{o}(1,8),\\
			P_{o}(2,1)&=P_{o}(2,2)=P_{o}(2,3)=P_{o}(2,4),\\
			P_{o}(2,5)&=P_{o}(2,6)=P_{o}(2,7)=P_{o}(2,8),\\
			P_{o}(3,m)&=P_{o}(3,m+4), \: m=1,\ldots,4.
		\end{split}
	\end{equation}
	
	\noindent The system's outage probability is given by calculating the stationary distribution of the MC and, based on the aforementioned relations, is equal to
	
	\begin{align}
		P_{out}(\gamma)&=\sum_{m=1}^{8}\boldsymbol{\pi}_{m}P_{o}(3,m)\nonumber\\
		&=P_{o}(1,1)P_{o}(2,1)+P_{o}(1,1)\left[1-P_{o}(2,1)\right] P_{o}(3,2)\nonumber\\
		&\quad+\left[1-P_{o}(1,1)\right]P_{o}(2,5)P_{o}(3,3)\nonumber\\
		&\quad+\left[1-P_{o}(1,1)\right]\left[1-P_{o}(2,5)\right]P_{o}(3,4),
	\end{align}
	where the probabilities $P_{o}(i,m)$ can be calculated by using the aforementioned outage expressions in \eqref{out_prob_node}, \eqref{cdf_exp} or \eqref{out_prob_saa}.
	
	Regarding the diversity gain of the network, we need to examine which are the dual-mode relays within the network. Thus, if all the relays are active i.e., $H=\emptyset$, then the proposed protocol can achieve diversity order equal to the number of hops, while for $H=\{1\}$ or $H=\{2\}$, from Corollary \ref{cor_div} the maximum diversity order is equal to one. Finally, if all the relays are dual-mode i.e., $H=\{1,2\}$, then at high $\SNR$ the outage probability converges to a floor value which is calculated from \eqref{floor_b} as $e(2,2)=q^{2}$.

\begin{table}[t!]\centering
	\caption{Summary of simulation parameters}\label{Table1}
	\scalebox{0.9}{
		\begin{tabular}{| l | l |}\hline
			\textbf{Description} & \textbf{Value}\\\hline
			Number of relays & $N=2$ or $3$\\\hline
			Distance between $S$ and $D$ & $d_{0,N+1}=3$ m\\\hline
			Path loss exponent & $\eta=2$\\\hline
			Variance of AWGN & $\sigma^{2}=1$\\\hline 
			$\SNR$ threshold & $\gamma=0$ dB\\\hline
			Probability of silent mode & $q=0.1$ or $0.5$\\\hline
	\end{tabular}}\vspace{-5mm}
\end{table}

	\section{Numerical Results} \label{results}
	In this section, we validate our theoretical analysis and main analytical results with computer simulations to demonstrate the effect of various network parameters on the performance of our proposed myopic protocol. For the simulations, the following parameters are considered. The distance between $S$ and $D$ is set to a fixed value of $d_{0,N+1}=3$ m and the distance between two consecutive nodes is the same for all nodes i.e., $d_{i,i+1}=d=d_{0,N+1}/(N+1)$ $\forall$ $i=0,\ldots,N$. In addition, the path loss exponent is equal to $\eta=2$, the variance of the AWGN is normalized to $\sigma^{2}=1$ and the energy of the channel coefficients is normalized to $\mathbb{E}\left[\left|h_{i,j}\right|^2\right]=1$. Finally, the $\SNR$ threshold is set to $\gamma=0$ dB. The values of the simulation parameters are summarized in Table \ref{Table1}. Note that, in all the presented figures, the analytical results are illustrated with lines (solid, dashed or dotted) and the simulation results with markers.
	
	It can be easily observed that the system's outage probability depends on how the transmit power of each node is divided. Therefore, our results are numerically optimized with respect to the power splitting parameters $a_{i,j}$. We formulate the system's outage probability minimization problem as follows\vspace{-1mm}
	\begin{mini}
		{{\scriptstyle\{a_{i,j}\}}}{P_{out}(\gamma)}{}{}
		\label{optim}
		\addConstraint{\sum_{j=i+1}^{L_{i}+i} a_{i,j}}{=1,}{i=0,\ldots,N}
		\addConstraint{0 < a_{i,j}}{< 1,\quad}{i=0,\ldots,N}
		\addConstraint{q=0.\:\:\:}{ }{ }
	\end{mini}
	Note that the second constraint ensures that a portion of the power $P$ will be allocated to all the available channel links. Moreover, the optimization is formulated for the case where all the relays of the network are always active i.e., $q=0$, which leads to a simpler implementation of the protocol, since the power management at each relay will not depend on the general knowledge of the communication modes of all the relays at each time-slot. Due to the complexity of the derived expressions, the aforementioned minimization problem can be solved using numerical tools, such as the \textit{NMinimize} function of Mathematica \cite{Mathematica}.
	
	\begin{figure}[t!]\centering
		\includegraphics[width=0.9\linewidth]{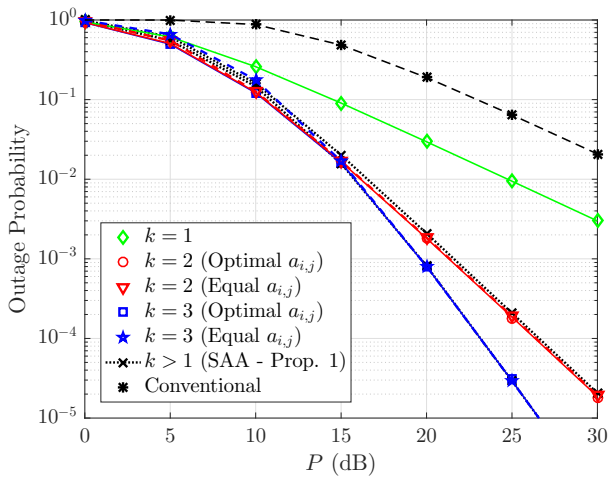}\vspace{-3mm}
		\caption{Outage probability versus $P$ for a network topology with $N=2$ relays, $k=1,2,3$ hops, $H=\emptyset$, $\gamma=0$ dB and $d=1$ m; the theoretical results are depicted with lines and the simulation results with markers.}
		\label{fig:out_prob_N_2_div}
	\end{figure}
	
	Figs. \ref{fig:out_prob_N_2_div} and \ref{fig:out_prob_N_3_div} illustrate the system's outage probability against the transmit power $P$ in a network setting with $N=2$ and $3$ relays, respectively, considering the case where the relays are always active. The outage performance is investigated for all possible $k$-hop myopic scenarios i.e., for $k=1,\dots,N+1$. For these scenarios, the conventional multi-hop DF scheme used in \cite{hasna2003}, where each node sends a signal only to its subsequent node through orthogonal channels, is used as a performance benchmark. The outage probability in this case is given by
	
	\begin{align}
	P_{out,c}&=1-\prod_{j=1}^{N+1}\mathbb{P}\left[\SNR_{j}\geq \gamma_{c}\right]\nonumber\\
	&=1-\exp\left[-\dfrac{\left(N+1\right)\gamma_{c}d^{\eta}\sigma^{2}}{P}\right],
	\end{align}

	\noindent where $\gamma_{c}=(\gamma+1)^{N+1}-1$. We observe that, in both figures, our proposed protocol is superior to the conventional multi-hop scheme. In particular, it can be seen that the $k$-hop myopic strategy outperforms the typical multi-hop DF scheme in terms of outage performance and coding gain. Moreover, for $k>1$ the myopic scheme also achieves a higher diversity gain. Furthermore, in Fig. \ref{fig:out_prob_N_2_div}, the optimized outage performance of the network, according to \eqref{optim}, is also compared to the case where the power at each transmitter is equally divided to each channel link i.e., $a_{i,j}=1/\min(k,N-i+1)$, $0\leq i \leq N$. It is observed that, while at high $\SNR$s the performance is almost the same for both cases, in the low $\SNR$ regime the outage probability is slightly improved for the numerically optimized values of $a_{i,j}$. It is also worth noting that, when $P$ is equally divided, the system cannot fully benefit from the increase on the number of hops, since the outage performance for $k=3$ is worse than the $2$-hop scenario at low transmit power values.
	
	\begin{figure}[t!]\centering
		\includegraphics[width=0.9\linewidth]{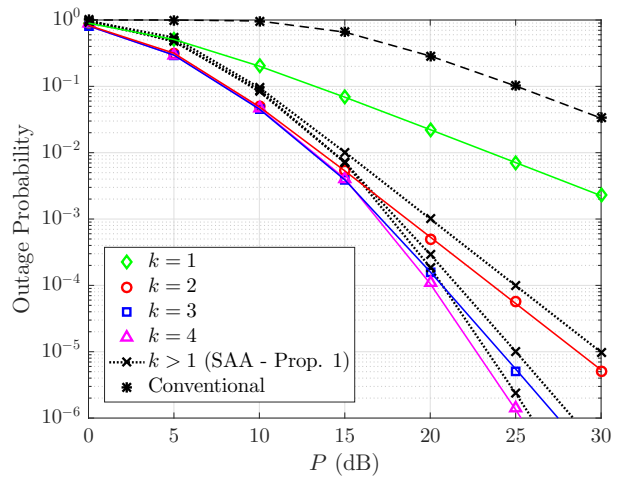}\vspace{-3mm}
		\caption{Outage probability versus $P$ for a network topology with $N=3$ relays, $k=1,2,3,4$ hops, $H=\emptyset$, $\gamma=0$ dB and $d=0.75$ m; the theoretical results are depicted with lines and the simulation results with markers.}
		\label{fig:out_prob_N_3_div}
	\end{figure}
	
	Fig. \ref{fig:out_prob_N_2_div} shows that an increase in the number of hops results in an improvement of the outage probability performance, with the cases of $k=1$ and $k=2$ revealing the most significant difference (i.e., about $8$ dB gain for an outage probability equal to $10^{-2}$). However, for a transmit power value up to $15$ dB the $2$-hop scenario has similar performance with the $3$-hop case. Therefore, for lower values of $P$ a topology with $k<N+1$ can have a performance close to the case of full cooperation; this observation corresponds to the results of \cite{ong2008}. In addition, it can be seen that as the number of hops increases the diversity gain is also improved, which complies with our analysis indicating a maximum diversity order equal to $k$ when the relays of the network are always active. Similar results are derived in Fig. \ref{fig:out_prob_N_3_div} for $N=3$ relays. It is worth noting that in this case the outage probability for our protocol is slightly improved, compared to the results in Fig. \ref{fig:out_prob_N_2_div}, due to the reduced path loss. On the other hand, the diversity gain remains the same, as it depends only on the number of hops. Finally, we observe that in both figures the theoretical values (lines) perfectly match to the simulation results (markers), which validates the accuracy of our analysis, while for $k>1$ the obtained SAA expression in Proposition \ref{saa_prop} provides a tight approximation of the actual performance of the network.
	
	\begin{figure}[t!]\centering
		\includegraphics[width=0.9\linewidth]{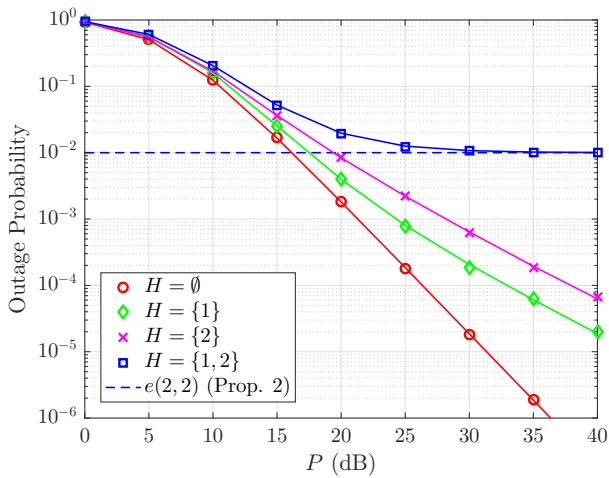}\vspace{-3mm}
		\caption{Outage probability versus $P$ for different sets of dual-mode relays; $N=2$ relays, $k=2$ hops, $q=0.1$, $\gamma=0$ dB and $d=1$ m; the theoretical results are depicted with lines and the simulation results with markers.}
		\label{fig:out_prob_N_2_H}
	\end{figure}
	
	\begin{figure}[t!]\centering
		\includegraphics[width=0.9\linewidth]{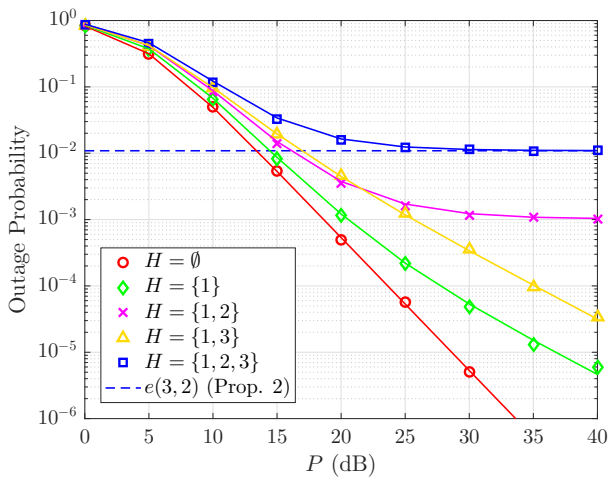}\vspace{-3mm}
		\caption{Outage probability versus $P$ for different sets of dual-mode relays; $N=3$ relays, $k=2$ hops, $q=0.1$, $\gamma=0$ dB and $d=0.75$ m.}
		\label{fig:out_prob_N_3_H}
	\end{figure}

	In Figs. \ref{fig:out_prob_N_2_H} and \ref{fig:out_prob_N_3_H} the achieved outage probability is presented for various sets of dual-mode relays, under the 2-hop myopic protocol and for a network topology with $N=2$ and $3$ relays, respectively. As expected, the case where all the relays are active is superior to any other scenario and is the only case that can achieve the maximum possible diversity gain $k=2$. On the other hand, if the network consists of only dual-mode relays, the outage performance is significantly deteriorated and as $P$ increases, the outage probability converges to the floor value calculated in Proposition 2. The performance of any other deployment scenario lies between the limits set by the previous two extreme cases.
	
	Specifically, in Fig. \ref{fig:out_prob_N_2_H} it is observed that both intermediate cases with $H=\{1\}$ and $H=\{2\}$ achieve the same diversity order, which is equal to one. However, when only $R_1$ is able to switch between active and silent mode, the system achieves a better outage performance than the system where only $R_2$ has dual-mode operations. This is expected, since $R_2$ is closer to $D$, which is also its only receiver, and so $R_2$ being silent has a higher impact to the system's outage probability. Again, in this figure we show that our simulation results (markers) validate our analysis (lines). In Fig. \ref{fig:out_prob_N_3_H}, it can be seen that if only $R_{1}$ is dual-mode, the achieved diversity order is decreased to one. Moreover, in the case of $H=\{1,3\}$ the same diversity order can be achieved, even though more dual-mode relays are used, since there is still a possible path from $S$ to $D$ which does not depend on the dual-mode operations ($S\rightarrow R_2 \rightarrow D$). On the contrary, the performance of the topology considering $H=\{1,2\}$ converges to an outage floor, due to the consecutive order of the two dual-mode relays. Therefore, the selection of which relays will operate in dual-mode is critical for the network's performance, and especially the achieved diversity gain, and this observation follows our discussion in Section \ref{dmt_sect}.

	Finally, Fig. \ref{fig:out_prob_Z_q} depicts the outage probability versus $P$ for different number of branches and various combinations of dual-mode relays, where each branch deploys a network setting with $N=2$ relays and $k=2$ hops. First of all, we can see that an increase in the number of branches results in the improvement of the system's outage probability. Specifically, compared to the single branch with only active relays scenario, doubling the branches will also double the achieved diversity gain, while if the second branch has only dual-mode relays, the system will not achieve higher diversity gain but its outage performance will still be enhanced. Similarly, if we consider a single branch network with only dual-mode relays, doubling the branches will significantly decrease the outage floor value. The aforementioned scenarios are also presented for different values of $q$ i.e., the probability of a dual-mode relay to be silent. As expected, as this probability increases the system's outage probability increases as well, since the dual-mode relays will remain silent for a larger number of time-slots. However, the system's performance remains unaffected in terms of diversity gain.
	
	\begin{figure}[t!]\centering
		\includegraphics[width=0.9\linewidth]{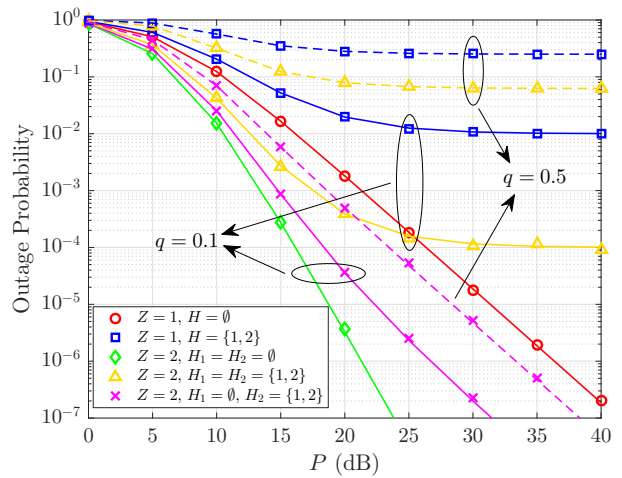}\vspace{-3mm}
		\caption{Outage probability versus $P$ for different number of branches; $N=2$ relays, $k=2$ hops, $\gamma=0$ dB and $d=1$ m.}
		\label{fig:out_prob_Z_q}
	\end{figure}

	\section{Conclusions}\label{conc}
	In this paper, we presented a new protocol over a multi-hop cooperative network, based on the myopic strategy, where the relays have buffers of finite size and can operate in two communication modes, namely the active and silent mode. A general methodology that captures how the contents of each relay’s buffer and the communication operation of each relay evolve with time was proposed by using a MC formulation. Under this framework, we derived the outage probability of the system and a general expression for the resulting DMT, and we investigated the maximum diversity order that can be achieved, based on different topology scenarios. The proposed protocol was also generalized for multi-branch networks. We demonstrated that as the number of hops in the proposed protocol increases, the system’s performance is enhanced in terms of outage probability, while the achieved diversity gain depends on both the number of hops and the group of relays that have dual-mode operations.
	\appendix
	\subsection{Proof of Lemma \ref{steady}} \label{proofA}
	We first need to verify that some properties of the state transition matrix $\mathbf{A}$ hold, in order to ensure that it has a unique stationary distribution. Specifically, $\mathbf{A}$ must be a column stochastic matrix, which is irreducible and aperiodic. A column stochastic matrix is a square matrix of non-negative terms in which the elements in each column sum up to one, while a non-negative matrix is called irreducible if every pair of states can communicate. Finally, the period of a state $s_{m}$ is the greatest common divisor of the set $\{t\in \mathbb{N}:p_{m,m}^{(t)}>0\}$ and if the period is $1$, the state is aperiodic. The transition matrix is aperiodic if all the states are aperiodic \cite{gallager2013}.

	For any MC under the proposed framework, the transitions from state $s_{m}$ to all the possible states $s_{l}$ have probabilities that sum up to one i.e., $\sum_{l=1}^{M}p_{l,m}=1$. Therefore, the transition matrix is column stochastic. Moreover, from the structure of the problem, it is observed that there is a path from any state to any other state of the MC. Consequently, all the states belong to a single communication class and the transition matrix is irreducible. An irreducible MC needs only one aperiodic state to imply that all states are aperiodic. In the defined MC the transition probability from state $s_{1}$ to the same state is always non-zero i.e., $p_{1,1}=\mathbb{P}(s_{1}\rightarrow s_{1})>0$; hence the state $s_{1}$ is aperiodic. Since there is only one communication class, all states are aperiodic and therefore, the transition matrix is aperiodic. As the required properties of the transition matrix hold, we conclude that it has a unique stationary distribution which is given as in \cite{krikidis2012}.
	\subsection{Proof of Theorem \ref{out_prob_thm}} \label{proofB}
	In the proposed protocol, at each time-slot, the $\SNR$ at the $j$-th receiver is calculated based on which nodes transmit a signal to the receiver, given the network state $s_{m}$. According to the presented theoretical framework, the $i$-th node transmit a signal to the $j$-th receiver, if \eqref{ind_fun} is equal to one i.e., the node is active and its specific buffer's element is not empty. Thus, the $\SNR$ of the $j$-th node for the transition $s_{m}\rightarrow s_{l}$ equals
	\begin{equation}
		\label{snr_proof}
		\SNR_{j,m}=\dfrac{P}{\sigma^{2}}\left(\sum_{i=[j-k]^{+}}^{j-1} |h_{i,j}|\sqrt{\dfrac{a_{i,j}}{d_{i,j}^{\eta}}}\mathds{1}_{(i\rightarrow j),m}\right)^2.
	\end{equation}
	The probability of having an outage event at the $j$-th node can be expressed as
	\begin{equation}
		\label{out_prob_proof}
		P_{o}(j,m)=\mathbb{P}\left[\sum_{i=[j-k]^{+}}^{j-1} |h_{i,j}|\sqrt{\dfrac{a_{i,j}}{d_{i,j}^{\eta}}}\mathds{1}_{(i\rightarrow j),m}<\sqrt{\dfrac{\gamma \sigma^{2}}{P}}\right].
	\end{equation}
	To simplify the analysis, the characteristic function approach is used for the derivation of the outage probability. Each term of the above sum follows a Rayleigh distribution, hence its characteristic function is \cite{papoulis2002}
	\begin{equation}
		\label{char_fun_proof}
		\phi_{i,j}(t)=1+\jmath t\sqrt{\dfrac{2\pi a_{i,j}}{d_{i,j}^{\eta}}}\exp\bigg(-\dfrac{a_{i,j}t^2}{2d_{i,j}^{\eta}}\bigg)\Phi\bigg(\jmath t\sqrt{\dfrac{a_{i,j}}{d_{i,j}^{\eta}}}\bigg).
	\end{equation}
	Since the sum is a linear combination of independent random variables, its characteristic function is the product of each individual's characteristic function. Using the Gil-Pelaez inversion theorem \cite{pelaez1951}, we can obtain $ P_{o}(j,m) $ as \eqref{out_prob_node}.
	\subsection{Proof of Theorem \ref{dmt_thm}} \label{proofC}
	By considering $P\rightarrow \infty$, we notice that the outage probability at the $j$-th relay converges to zero if it receives a signal from at least one transmitter i.e., $C_{j,m}>0$, otherwise it is equal to one. Thus, the relays that receive a signal are almost surely able to decode the message. Moreover, if a transmitter has a full buffer at every time-slot and is always active, then its corresponding receivers will also have constantly a full buffer. For example, a relay that receives information from $S$ will always have a full buffer, and if this relay is not in dual-mode operation, then its corresponding receivers will also have full buffers at every time-slot. This implies that the steady states of the transition matrix that do not conform with the above observations will converge to zero.

	According to \eqref{out_prob_sys}, the system's outage probability is a sum of $M$ terms that depends on the stationary distribution of the MC and the outage probability at the destination for each steady state. At the high $\SNR$ regime, the terms of the sum in \eqref{out_prob_sys} that correspond to steady states converging to zero are omitted. For each of the remaining terms, we obtain the network instance associated with the corresponding MC state $s_m$ and derive the flow graph $G_m$ from $S$ to $D$ for the paths that will surely forward the signals to $D$ i.e., by omitting the dual-mode relays. Let $\mathcal{T}_m=\{T_{1,m}, T_{2,m}, \ldots\}$ denote the complete set of cuts separating $S$ and $D$ at $G_m$, and $\mathcal{L}_{i,m}$ denote the number of links crossing the cut $T_{i,m}$. According to the information theoretic max-flow min-cut theorem \cite{cover1991}, the DMT that can be achieved by $G_m$ is upper bounded by
	\begin{equation}
		\delta_m(\rho)\leq\mathcal{L}_m=\min (\mathcal{L}_{i,m}).
	\end{equation}
	Note that $\mathcal{L}_m$ is equal to the minimum number of edge-disjoint paths that forward information from $S$ to $D$. The system's outage probability is dominated by the terms that have the lowest order, or equivalently by the minimum number of links crossing a cut $T_{i,m}$ from all the derived flow graphs $G_m$, $1\leq m \leq M$, $\boldsymbol{\pi}_{m}\nrightarrow 0$. Thus, by defining the target data rate as $R=\log{(1+\gamma)}=\rho\log{P}$ and based on the results of \cite{sreeram2012} about the DMT of networks that can organize their relays into parallel paths, the DMT is given by \eqref{dmt_gen}.

	\subsection{Proof of Proposition \ref{floor_prop}}\label{proofD}
	Recall that, by considering $P\rightarrow \infty$, the outage probability of a node converges to zero if it receives at least one signal during one time-slot. Thus, a receiver is in outage only if all its corresponding transmitters remain silent. For an arbitrary signal $x(t)$ the system is in outage if the signal cannot reach $D$, so the outage floor value is given by the aggregate probability of the instances that fail to transfer the signal to $D$. According to the proposed protocol, until time-slot $(t+k-1)$ the first $k$ relays will definitely receive and decode the signal $x(t)$ from $S$, which is always active. If $N=k$, then the next node is $D$ and an outage occurs if all the relays remain silent, for which the probability is calculated as \eqref{floor_b}.
	
	For $N>k$, relay $k+1$ cannot receive $x(t)$ if the previous $k$ relays remain silent during time-slot $(t+k)$, which occurs with probability $q^{k}$. For the remaining time-slots, we can calculate the probability of the system being in outage by equivalently considering the outage probability of another network topology with $N-2$ relays and $k-1$ hops, since the first relay cannot contribute anymore to the transmission of $x(t)$ and relay $k+1$ has not decoded the signal, and so they can be omitted. On the contrary, relay $k+1$ receives $x(t)$ at time-slot $(t+k)$ if at least one of its transmitters is active, and this occurs with probability $(1-q^{k})$. In this case, the outage probability for the remaining time-slots is matched to the outage probability of an equivalent topology with $N-1$ relays and $k$ hops, by considering the first relay as the new source node. This recursive behavior leads to the calculation of the outage floor by \eqref{floor_a}.


\begin{thebibliography}{10}
	\bibitem{nicolaides2020} A. Nicolaides, C. Psomas, and I. Krikidis, ``Outage analysis of myopic multi-hop relaying: A Markov chain approach,'' in \emph{Proc. IEEE Global
	Commun. Conf.,} Taipei, Taiwan, Dec. 2020, pp. 1--6.

	\bibitem{saad2020} W. Saad, M. Bennis, and M. Chen, ``A vision of 6G wireless systems: Applications, trends, technologies, and open research problems,'' \emph{IEEE Netw.,} vol. 34, no. 3, pp. 134--142, May 2020.
	
	\bibitem{tehrani2014} M. N. Tehrani, M. Uysal, and H. Yanikomeroglu, ``Device-to-device communication in 5G cellular networks: challenges, solutions, and future directions,'' \emph{IEEE Commun. Mag.,} vol. 52, no. 5, pp. 86--92, May 2014.
	
	\bibitem{chen2021} X. Chen, D. W. K. Ng, W. Yu, E. G. Larsson, N. Al-Dhahir, and R. Schober, ``Massive access for 5G and beyond,'' \emph{IEEE J. Sel. Areas Commun.,} vol. 39, no. 3, pp. 615--637, Mar. 2021.
	
	\bibitem{atallah2016} R. Atallah, M. Khabbaz, and C. Assi, ``Multihop V2I communications: a feasibility study, modeling, and performance analysis,'' \emph{IEEE Trans. Veh. Technol.,} vol. 66, no. 3, pp. 2801--2810, Mar. 2017.
	
	\bibitem{rois2015} J. G. Rois, F. G. Cuba, R. M. Akdeniz, F. J. G. Castaño, J. C. Burguillo, S. Rangan, and B. Lorenzo, ``On the analysis of scheduling in dynamic duplex multihop mmWave cellular systems,'' \emph{IEEE Trans. Wireless Commun.,} vol. 14, no. 11, pp. 6028--6042, Nov. 2015.
	
	\bibitem{saeed2019} N. Saeed, A. Celik, M. S. Alouini, and T. Y. Al-Naffouri, ``Performance analysis of connectivity and localization in multi-hop underwater optical wireless sensor networks,'' \emph{IEEE Trans. Mobile Comput.,} vol. 18, no. 11, pp. 2604--2615, Nov. 2019.
	
	\bibitem{chen2018} Y. Chen, N. Zhao, Z. Ding, and M. S. Alouini, ``Multiple UAVs as relays: multi-hop single link versus multiple dual-hop links," \emph{IEEE Trans. Wireless Commun.,} vol. 17, no. 9, pp. 6348--6359, Sep. 2018.
	
	\bibitem{wang2020} H. Wang, Y. Zhang, X. Zhang and Z. Li, ``Secrecy and covert communications against UAV surveillance via multi-hop networks,'' \emph{IEEE Trans. Commun.,} vol. 68, no. 1, pp. 389--401, Jan. 2020.
	
	\bibitem{hasna2004} M. O. Hasna and M. S. Alouini, ``A performance study of dual-hop transmissions with fixed gain relays,'' \emph{IEEE Trans. Wireless Commun.,} vol. 3, no. 6, pp. 1963--1968, Nov. 2004.
	
	\bibitem{bjornson2013} E. Bjornson, M. Matthaiou, and M. Debbah, ``A new look at dual-hop relaying: Performance limits with hardware impairments,'' \emph{IEEE Trans. Commun.,} vol. 61, no. 11, pp. 4512--4525, Nov. 2013.
	
	\bibitem{hasna2003} M. O. Hasna and M. S. Alouini, ``Outage probability of multihop transmission over Nakagami fading channels,'' \emph{IEEE Commun. Lett.,} vol. 7, no. 5, pp. 216--218, May 2003.
	
	\bibitem{jamali2015} V. Jamali, N. Zlatanov, H. Shoukry, and R. Schober, ``Achievable rate of the half-duplex multi-hop buffer-aided relay channel with block fading,'' \emph{IEEE Trans. Wireless Commun.,} vol. 14, no. 11, pp. 6240--6256, Nov. 2015.
	
	\bibitem{lin2015} X. Lin and J. G. Andrews, ``Connectivity of millimeter wave networks with multi-hop relaying,'' \emph{IEEE Wireless Commun. Lett.}, vol. 4, no. 2, pp. 209--212, Apr. 2015.
	
	\bibitem{laneman2004} J. N. Laneman, D. N. C. Tse, and G. W. Wornell, ``Cooperative diversity in wireless networks: Efficient protocols and outage behavior,'' \emph{IEEE Trans. Inf. Theory,} vol. 50, no. 12, pp. 3062--3080, Dec. 2004.
	
	\bibitem{ribeiro2005} A. Ribeiro, X. Cai, and G. B. Giannakis, ``Symbol error probabilities for general cooperative links,'' \emph{IEEE Trans. Wireless Commun.,} vol. 4, no. 3, pp. 1264--1273, May 2005.
	
	\bibitem{boyer2004} J. Boyer, D. D. Falconer, and H. Yanikomeroglu, ``Multihop diversity in wireless relaying channels,'' \emph{IEEE Trans. Commun.,} vol. 52, no. 10, pp. 1820--1830, Oct. 2004.
	
	\bibitem{sadek2007} A. K. Sadek, W. Su, and K. J. R. Liu, ``Multinode cooperative communications in wireless networks,'' \emph{IEEE Trans. Signal Process.,} vol. 55, no. 1, pp. 341--355, Jan. 2007.
	
	\bibitem{dong2012} C. Dong, L. L. Yang, and L. Hanzo, ``Performance analysis of multihop-diversity-aided multihop links,'' \emph{IEEE Trans. Veh. Tech.,} vol. 61, no. 6, pp. 2504--2516, Jul. 2012.
	
	\bibitem{sreeram2012} K. Sreeram, S. Birenjith, and P. V. Kumar, ``DMT of multihop networks: End points and computational tools,'' \emph{IEEE Trans. Inf. Theory,} vol. 58, no. 2, pp. 804--819, Feb. 2012.
	
	\bibitem{ong2008} L. Ong and M. Motani, ``Myopic coding in multiterminal networks,'' \emph{IEEE Trans. Inf. Theory,} vol. 54, no. 7, pp. 3295--3314, Jul. 2008.
	
	\bibitem{niyato2007} D. Niyato, E. Hossain, and A. Fallahi, ``Sleep and wakeup strategies in solar-powered wireless sensor/mesh networks: Performance analysis and optimization,'' \emph{IEEE Trans. Mobile Comput.,} vol. 6, no. 2, pp. 221--236, Feb. 2007.
	
	\bibitem{medepally2010} B. Medepally and N. B. Mehta, ``Voluntary energy harvesting relays and selection in cooperative wireless networks,'' \emph{IEEE Trans. Wireless Commun.,} vol. 9, no. 11, pp. 3543--3553, Nov. 2010.
	
	\bibitem{luo2013} Y. Luo, J. Zhang, and K. B. Letaief, ``Relay selection for energy harvesting cooperative communication systems,'' in \emph{Proc. IEEE Global Commun. Conf.,} Atlanta, USA, Dec. 2013, pp. 2514--2519.
	
	\bibitem{morsi2018} R. Morsi, D. S. Michalopoulos, and R. Schober, ``Performance analysis of near-optimal energy buffer aided wireless powered communication,'' \emph{IEEE Trans. Wireless Commun.,} vol. 17, no. 2, pp. 863--881, Feb. 2018.
	
	\bibitem{wang2015} S. Wang, W. Guo, Z. Zhou, Y. Wu, and X. Chu, ``Outage probability for multi-hop D2D communications with shortest path routing," \emph{IEEE Commun. Lett.,} vol. 19, no. 11, pp. 1997--2000, Nov. 2015.
	
	\bibitem{korpi2017} D. Korpi, M. Heino, C. Icheln, K. Haneda, and M. Valkama, ``Compact inband full-duplex relays with beyond 100 dB self-interference suppression: enabling techniques and field measurements,'' \emph{IEEE Trans. Antennas Propag.,} vol. 65, no. 2, pp. 960--965, Feb. 2017.
	
	\bibitem{li2016} T. Li, P. Fan, and K. B. Letaief, ``Outage probability of energy harvesting relay-aided cooperative networks over rayleigh fading channel,'' \emph{IEEE Trans. Veh. Technol.,} vol. 65, no. 2, pp. 972--978, Feb. 2016.
	
	\bibitem{hu2005} J. Hu and N. C. Beaulieu, ``Accurate simple closed-form approximations to Rayleigh sum distributions and densities,'' \emph{IEEE Commun. Lett.,} vol. 9, no. 2, pp. 109--111, Feb. 2005.
	
	\bibitem{hitczenko1998} P. Hitczenko, ``A note on a distribution of weighted sums of i.i.d. Rayleigh random variables,'' \emph{Sankhyā: The Indian Journal of Statistics, Series A (1961-2002),} vol. 60, no. 2, pp. 171--175, Jun. 1998.
	
	\bibitem{zheng2003} L. Zheng and D. N. C. Tse, ``Diversity and multiplexing: a fundamental tradeoff in multiple-antenna channels,'' \emph{IEEE Trans. Inf. Theory,} vol. 49, no. 5, pp. 1073--1096, May 2003.
	
	\bibitem{goldsmith2005} A. Goldsmith, \emph{Wireless Communications.} Cambridge: Cambridge University Press, 2005.
	
	\bibitem{Mathematica} Wolfram Mathematica Documentation - NMinimize. [Online]. Available: https://reference.wolfram.com/language/ref/NMinimize.html
	
	\bibitem{gallager2013} R. G. Gallager, \emph{Stochastic Processes: Theory for Applications.} Cambridge: Cambridge University Press, 2013.
	
	\bibitem{krikidis2012} I. Krikidis, T. Charalambous, and J. S. Thompson, ``Buffer-aided relay selection for cooperative diversity systems without delay constraints,'' \emph{IEEE Trans. Wireless Commun.,} vol. 11, no. 5, pp. 1957--1967, May 2012.
	
	\bibitem{papoulis2002} A. Papoulis and S. U. Pillai, \emph{Probability, Random Variables, and Stochastic Processes,} 4th ed. McGraw Hill, 2002.
	
	\bibitem{pelaez1951} J. Gil-Pelaez, ``Note on the inversion theorem,'' \emph{Biometrika,} vol. 38, no. 3-4, pp. 481--482, Dec. 1951.
	
	\bibitem{cover1991} T. M. Cover and J. Thomas, \emph{Elements of Information Theory.} New York: Wiley, 1991.
\end{thebibliography}
\end{document}